\documentclass[12pt]{amsart}
\usepackage{enumerate}
\usepackage{hyperref}
\usepackage{mathabx}
\usepackage{blkarray}
\usepackage{tikz} 
\usetikzlibrary{backgrounds}
\usetikzlibrary{shapes}
\usetikzlibrary{decorations.markings}
\usetikzlibrary{arrows,automata}
\usepackage{subfig}
\usetikzlibrary{snakes}

\usepackage{algorithmicx,algorithm}
\usepackage{algpseudocode}
\hypersetup{
    pdftitle=   {A polynomial kernel for block graph deletion},
   pdfauthor=  {Eun Jung Kim and O-joung Kwon}
}
\newcommand\abs[1]{\lvert #1\rvert}

\usepackage{amsmath}
\newtheorem{THM}{Theorem}[section]
\newtheorem{LEM}[THM]{Lemma}
\newtheorem{COR}[THM]{Corollary}
\newtheorem{PROP}[THM]{Proposition}

\newtheorem{RULE}{Reduction Rule}

\theoremstyle{remark}

\theoremstyle{definition}

\def\K_#1{{K_{#1}}}
\def\S_#1{\overline{K_{#1}}}

\newcommand\cO{\mathcal{O}}

\newcommand\cal{\mathcal}

\newcommand{\BLOC}{\textsc{Block Graph Deletion}}
\newcommand{\disjointBLOC}{\textsc{Disjoint Block Graph Deletion}}

\newcommand{\YES}{\textsc{Yes}}
\newcommand{\NO}{\textsc{No}}

\begin{document}
\title[Block graph deletion]{A polynomial kernel for block graph deletion}
\author{Eun Jung Kim}
\address{LAMSADE, CNRS - Universit\'e Paris Dauphine, France}
\email{eunjungkim78@gmail.com}
\author{O-joung Kwon} 
\address{Institute for Computer Science and Control, Hungarian Academy of Sciences, Hungary}
\email{ojoungkwon@gmail.com}
\thanks{The second author is supported by ERC Starting Grant PARAMTIGHT (No. 280152).}
\thanks{An extended abstract appeared in 
  Proc. 10th International Symposium on Parameterized and Exact Computations, 2015~\cite{KimK2015}.}
\date{\today}
\begin{abstract}
In the \textsc{Block Graph Deletion} problem, we are given a graph $G$ on $n$ vertices and a positive integer $k$, and the objective is to check whether it is possible to delete at most $k$ vertices from $G$ to make it a block graph, i.e., a graph in which each block is a clique.
In this paper, we obtain a kernel with $\mathcal{O}(k^{6})$ vertices for the \textsc{Block Graph Deletion} problem.
This is a first step to investigate polynomial kernels for deletion problems into non-trivial classes of graphs of bounded rank-width, but unbounded tree-width. Our result also implies that \textsc{Chordal Vertex Deletion} admits a polynomial-size kernel on diamond-free graphs.
For the kernelization and its analysis, we introduce the notion of `complete degree' of a vertex. We believe that the underlying idea can be potentially applied to other problems. 
We also prove that the \textsc{Block Graph Deletion} problem can be solved in time $10^{k}\cdot n^{\mathcal{O}(1)}$.

\end{abstract}
\maketitle

\section{Introduction}\label{sec:introduction}

In parameterized complexity, an instance of a parameterized problem consists in a pair $(x,k)$, where $k$ is a secondary measurement, called the \emph{parameter}. 
A parameterized problem $Q\subseteq \Sigma^* \times N$ is \emph{fixed-parameter tractable} (\emph{FPT}) if there is an algorithm which decides whether $(x,k)$ belongs to $Q$ in time $f(k)\cdot \abs{x}^{\mathcal{O}(1)}$ for some computable function $f$. Such an algorithm is called an {\em FPT} algorithm. We call an FPT algorithm a {\em single-exponential} FPT algorithm if it runs in time $c^k\cdot \abs{x}^{\mathcal{O}(1)}$ for some constant $c$.
A parameterized problem is said to admit a \emph{polynomial kernel} if there is a polynomial time algorithm in $\abs{x}+k$, called a \emph{kernelization algorithm}, that reduces an input instance into an instance with size bounded by a polynomial function in $k$, while preserving the \YES/\NO\ answer.

Graph modification problems constitute a fundamental class of graph optimization problems. 
Typically, for a class $\Phi$ of graphs, a set $\Psi$ of graph operations  and a positive integer $k$, 
we want to know whether it is possible to transform an input graph into a graph in $\Phi$ by at most $k$ operations chosen in $\Psi$.
One of the most intensively studied graph modification problems is the \textsc{Feedback Vertex Set} problem. 
Given a graph $G$ and an integer $k$ as input, the \textsc{Feedback Vertex Set} problem asks whether $G$ has a vertex subset of size at most $k$ whose removal makes it a forest, which is a graph without cycles. 
The \textsc{Feedback Vertex Set} problem is known to admit an FPT algorithm~\cite{Bodlaender91, DowneyF93} and the running time has been subsequently improved by a series of papers~\cite{RamanSS02, KanjPS04, GuoGHNW06, DehneFLR07, ChenFLL08, CaoCL10, CyganNPP11,Kociumaka14}. Also, Thomass\'{e}~\cite{Thomasse2009} showed that it admits a kernel on $O(k^2)$ vertices. 

The \textsc{Feedback Vertex Set} problem has been generalized to deletion problems for more general graph classes. \emph{Tree-width}~\cite{RS1986a} is one of the basic parameters in graph algorithms and plays an important role in structural graph theory. Since forests are exactly the graphs of tree-width at most $1$, the natural question is to decide, for an integer $w\ge 2$, whether there is an FPT algorithm with parameter $k$ to find a vertex subset of size at most $k$ whose removal makes it a graph of tree-width at most $w$ (called \textsc{Tree-width $w$ Vertex Deletion}).
Courcelle's meta theorem~\cite{Courcelle90} implies that the \textsc{Tree-width $w$ Vertex Deletion} is FPT. Recently it is proved to admit a single-exponential FPT algorithm and a (non-uniform) polynomial kernel (a kernel of size $\mathcal{O}(k^{g(w)})$ for some function $g$)~\cite{FominLMS12, KLPRRSS13}.

On the other hand, there are interesting open questions related to two natural graph classes having tree-like structures. A graph is \emph{chordal} if it does not contain any induced cycle of length at least $4$. Chordal graphs are close to forests as a forest is a chordal graph without triangles. Marx~\cite{Marx2010} firstly showed that the \textsc{Chordal Vertex Deletion} problem is FPT, and Cao and Marx~\cite{CaoMc14} improved that it can be solved in time $2^{\mathcal{O}(k\log k)} \cdot n^{\mathcal{O}(1)}$. However, it remains open whether there is a single-exponential FPT algorithm or a polynomial kernel~\cite{Marx2010, CaoMc14}. Another interesting class is the class of \emph{distance-hereditary graphs}, also known as graphs of \emph{rank-width} at most $1$~\cite{Oum05}. As many problems are tractable on graphs of bounded rank-width by the meta-theorem on graphs of bounded rank-width (equivalently, bounded clique-width)~\cite{CourcelleMR00}, it is worth studying the general \textsc{Rank-width $w$ Vertex Deletion} problem. Again, it is known to be FPT from the meta-theorem on graphs of bounded rank-width~\cite{CourcelleMR00}, but for our knowledge, it is open whether there is a single exponential FPT algorithm or a polynomial kernel for this problem even for $w=1$.

\emph{Block graphs} lie in the intersection of chordal graphs and distance-hereditary graphs, and they contain all forests. 
A graph is a {\em block graph} if each block of it forms a clique. It is not difficult to see that block graphs are exactly those not containing an induced cycle of length at least $4$ and a diamond (i.e. a cycle of length $4$ with a single chord) as an induced subgraph. We study the following parameterized problem.

\smallskip
\noindent
\fbox{\parbox{0.97\textwidth}{
\textsc{Block Graph Deletion}\\
\textbf{Input :} A graph $G$, an integer $k$ \\
\textbf{Parameter :} $k$ \\
\textbf{Question :} Is there a vertex subset $S$ of $G$ with $\abs{S}\le k$ such that $G-S$ is a block graph? }}
\vskip 0.2cm

Our main results are stated in the next two theorems.

\begin{THM}\label{thm:main1}
The \BLOC\ admits a kernel with $O(k^{6})$ vertices.
\end{THM}

\begin{THM}\label{thm:main2}
The \BLOC\ can be solved in time $10^k\cdot n^{\mathcal{O}(1)}$.
\end{THM}

Our kernelization is motivated by the quadratic vertex-kernel by Thomass\'{e}~\cite{Thomasse2009}. In~\cite{Thomasse2009}, basic reduction rules are applied so that whenever the size of the instance is still large, there must be a vertex of large degree (otherwise, it is a \NO-instance). Then a vertex $v$ of large degree witnesses either so-called the {\em sunflower} structure, or the {\em 2-expansion} structure. Our kernelization employs a similar strategy. In order to work with block graphs instead of forests, we come up with the notion of the {\em complete degree} of a vertex, which replaces the role of the usual degree of a vertex in {\sc Feedback Vertex Set}. Also, we need to bound the size of a block which might appear in a block graph $G-S$, if such a set $S$ of size at most $k$ exists. Our single-exponential algorithm is surprisingly analogous to the algorithm of Chen. {\em et al.}~\cite{ChenFLL08} for {\sc Feedback Vertex Set} although the analysis is non-trivial. 

Since block graphs are exactly diamond-free chordal graphs, we have the following as a corollary of Theorem~\ref{thm:main1} and Theorem~\ref{thm:main2}.

\begin{COR}\label{corr:chordal}
On diamond-free graphs, {\sc Chordal Vertex Deletion} admits a kernel with $O(k^{6})$ vertices and can be solved in time $O(10^k\cdot n^{\mathcal{O}(1)})$.
\end{COR}

\textbf{Update.} 
After this paper was presented at IPEC 2015, Agrawal et al.~\cite{AgrawalKLS2015} announced improvements of both results in the paper. 
Based on all of our reduction rules, they obtained a kernel with $\cO (k^4)$ vertices using a $4$-approximation algorithm for \BLOC.
For an FPT algorithm, they developed an $3.618^k \cdot n^{\mathcal{O}(1)}$ time algorithm for \textsc{Weighted Feedback Vertex Set}, and 
using a reduction from \BLOC\ to \textsc{Weighted Feedback Vertex Set} on graphs with no induced cycle of length $4$ and the diamond, they obtained an $4^k \cdot n^{\mathcal{O}(1)}$ time algorithm for the problem.

\section{Preliminaries}

All graphs considered in this paper are undirected and simple (without loops and parallel edges). 
For a graph $G$, we denote by $V(G)$ and $E(G)$ the vertex set and the edge set of $G$, respectively.
When we analyze the running time of an algorithm, we agree that $n=\abs{V(G)}$.

Given a graph $G$, a vertex $u$ is a {\em neighbor} of a vertex $v$ if $uv \in E(G)$. The {\em neighborhood} of a vertex set $X$ in $G$ is the set $\{u\in V(G):uv \in E(G) \text{~for some~}v\in X\}$ and denoted as $N_G(X)$, or simply $N(X)$. If $X$ consists of a single vertex $x$, then we write $N_G(\{x\})$ as $N_G(x)$. For two vertex sets $X,Y \subseteq V(G)$, we refer to the set $X\cap N_G(Y)$ by $N_X(Y)$. For $X\subseteq V(G)$, the set of vertices in $X$ having a neighbor in $V(G)\setminus X$ is denoted as $\partial_G(X)$. For $X\subseteq V(G)$, the graph obtained by deleting the vertices $X$ from $G$ is written as $G-X$. The same applies to an edge set. When $X$ is a single vertex $x$ or an edge $e$, we simply write $G-x$ and $G-e$, respectively. 
 A vertex $v$ of $G$ is called a {\em cut vertex} if the removal of $v$ from $G$ strictly increases the number of connected components. 
 A maximal connected subgraph of a graph without a cut vertex is called a {\em block} of it.
 Note than an edge can be a block.
A graph $G$ is \emph{$2$-connected} if $\abs{V(G)}\ge 3$ and it has no cut vertex.

A {\em block tree} ${\cal T}_G$ of a graph $G$ is the graph having ${\cal B}\cup {\cal C}$ as the vertex set, where ${\cal B}$ is the set of all blocks of $G$ and ${\cal C}$ is the set of all cut vertices of $G$, and there is an edge $Bc\in E({\cal T}_G)$ between $B\in {\cal B}$ and $c\in {\cal C}$ if and only if the cut vertex $c$ belongs to the block $B$ in $G$. The constructed graph does not contain a cycle. We say that a graph is a \emph{block graph obstruction}, or simply an {\em obstruction}, if it is isomorphic to a diamond, or an induced cycle $C_{\ell}$ of length $\ell$ for some $\ell\geq 4$. A vertex is {\em simplicial} in $G$ if $N_G(v)$ is a complete graph.

\section{Complete degree of a vertex}
We define a concept called the \emph{complete degree} of a vertex in a graph.
The definition of the complete degree is motivated by the following lemma, whose proof is deferred at the end of this section.

\begin{PROP}\label{prop:generalcompletedegree}
Let $G$ be a graph and let $v\in V(G)$ and let $k$ be a positive integer.
Then in $\mathcal{O}(kn^3)$ time, we can find either
\begin{enumerate}
\item $k+1$ obstructions that are pairwise vertex-disjoint, or
\item $k+1$ obstructions whose pairwise intersections are exactly the vertex $v$, or
\item $S_v\subseteq V(G)$ with $\abs{S_v}\le 7k$ such that $G-S_v$ has no obstruction containing $v$. 
\end{enumerate}
\end{PROP}

For a graph $G$ and $v\in V(G)$ such that $G$ has no $k+1$ vertex-disjoint obstructions and has no $k+1$ obstructions whose pairwise intersections are exactly the vertex $v$, 
the \emph{complete degree} of $v$ is defined as the minimum number of components of $G- (S_v\cup \{v\})$ among all possible $S_v\subseteq V(G)\setminus \{v\}$ where
\begin{itemize}
\item $\abs{S_v}\le 7k$, and
\item $G-S_v$ has no block graph obstruction containing $v$.
\end{itemize}
Note that if $G-S_v$ has no block graph obstruction containing $v$, then $G[N_G(v)\setminus S_v]$ is a disjoint union of complete graphs.

To prove Proposition~\ref{prop:generalcompletedegree}, we use the Gallai's $A$-path theorem.
For a graph $G$ and $A\subseteq V(G)$, 
an \emph{$A$-path} of $G$ is a path of length at least $1$ whose end vertices are in $A$, and all internal vertices
are in $V(G)\setminus A$. 

\begin{THM}[Gallai~\cite{Gallai1961}]\label{thm:gallaitheorem}
Let $G$ be a graph and let $A\subseteq V(G)$ and let $k$ be a positive integer.
Then, in $\mathcal{O}(kn^2)$ time, we can find either
\begin{enumerate}
\item $k+1$ vertex-disjoint $A$-paths, or
\item $X\subseteq V(G)$ with $\abs{X}\le 2k$ such that  $G- X$ has no $A$-paths.
\end{enumerate}
\end{THM}

\begin{proof}[Proof of Proposition~\ref{prop:generalcompletedegree}] 
Let $G_1:=(G-v)- E(G[N_G(v)])$. By Theorem~\ref{thm:gallaitheorem}, 
 we can find in time $\mathcal{O}(kn^2)$ either
\begin{enumerate}
\item $2k+1$ vertex-disjoint $N_G(v)$-paths in $G_1$, or
\item $X\subseteq V(G)$ with $\abs{X}\le 4k$ such that  $G_1- X$ has no $N_G(v)$-paths.
\end{enumerate}
Suppose that $G_1$ contains at least $2k+1$ pairwise vertex-disjoint $N_G(v)$-paths. Let $P$ be one of these $N_G(v)$-paths in $G_1$ with $p$ and $q$ as its end vertices, and let $P'$ be a shortest $p,q$-path in $G_1[V(P)]$. Note that $P'$ has length at least $2$. If $P'$ has length $2$, then $G[\{v\}\cup V(P')]$ is isomorphic to either $C_4$ or the diamond depending on the adjacency between $p$ and $q$ in $G$. If $P'$ has length at least $3$ and $pq\in E(G)$, then $G[V(P')]$ is an induced cycle of length at least $4$. If $P'$ has length at least $3$ and $pq\notin E(G)$, then $G[\{v\}\cup V(P')]$ is an induced cycle of length at least $5$. Thus, $G[\{v\}\cup V(P)]$ contains an obstruction, and $G$ contains either disjoint $k+1$ obstructions, or $k+1$ obstructions whose pairwise intersections are exactly $v$. 

So, we may assume that there exists $X\subseteq V(G_1)$ with $\abs{X}\le 4k$ such that $G_1- X$ has no $N_G(v)$-paths. Now, we greedily find a maximal set $\mathcal{P}$ of vertex-disjoint induced $P_3$ in $G[N_G(v)]$ by searching vertex subsets of size $3$. 
If there are $k+1$ vertex-disjoint induced $P_3$'s, then $G$ has $k+1$ diamonds whose pairwise intersections are exactly $v$. Otherwise, we set $S_v=X\cup \bigcup_{P\in \mathcal{P}} V(P)$ and notice that $\abs{S_v}\leq 7k$. Observe that $G-S_v$ has no block graph obstruction containing $v$.
Clearly, we can find $\mathcal{P}$ in time $\mathcal{O}(kn^3)$.
\end{proof}

In our algorithm, we need to find a vertex of sufficiently large complete degree and the corresponding deletion set $S_v$ in polynomial time. However, we just need sufficiently many complete graphs on the neighborhood, and do not need to compute the complete degree of each vertex exactly. The following lemma will be used to analyze the difference between an optimal set and an arbitrary set $S_v$ obtained by Proposition~\ref{prop:generalcompletedegree}.
\begin{LEM}\label{lem:countingcomponents}
Let $G$ be a graph and let $S_1, S_2\subseteq V(G)$ such that for each $1\le i\le 2$, $G- S_i$ is a disjoint union of complete graphs.
If $\abs{S_2}\le k$, then the number of components of $G- S_2$ is at least the number of components of $G- S_1$ minus $k$.
\end{LEM}
\begin{proof}
Note that $S_2$ can only remove at most $k$ vertices from the components of $G-S_1$, and two disjoint complete graphs cannot be merged into one complete graph by adding some new vertices. Thus, the number of components of $G- S_2$ is at least the number of components of $G- S_1$ minus $k$.
\end{proof}

\section{Finding a vertex of large complete degree}\label{sec:completeneighbor}

In this section, we prove that if a graph is reduced under certain rules and its size is still large, then 
there should exist a vertex of large complete degree.
To do this, we first provide basic reduction rules.

\subsection{Basic reduction rules}
\begin{RULE}[Block component rule]\label{rule:blockcomponent} 
If $G$ has a component $H$ that is a block graph, 
then we remove $H$ from $G$.
\end{RULE}

\begin{RULE}[Cut vertex rule]\label{rule:cutvertex} 
Let $v$ be a vertex of $G$ such that $G-v$ contains a component $H$ where $G[V(H)\cup \{v\}]$ is a connected block graph. 
Then we remove $H$ from $G$.
\end{RULE}

Two vertices $v, w$ in a graph $G$ are called \emph{true twins} if $N_G(v)\setminus \{w\}=N_G(w)\setminus \{v\}$ and $vw\in E(G)$.
Note that two simplicial vertices in a block of a block graph are true twins.
\begin{RULE}[Twin rule]\label{rule:twinreduction}
Let $S$ be the set of vertices  that are pairwise true twins in $G$.
If $\abs{S}\ge k+2$, then we remove vertices except $k+1$ vertices.
\end{RULE}

It is not hard to observe that Rules~\ref{rule:blockcomponent}, \ref{rule:cutvertex}, and \ref{rule:twinreduction} are sound.
Note that we can test whether a given graph is a block graph in quadratic time using an algorithm to partition the graph into blocks~\cite{HopcroftT1973}, and testing whether each block is a complete graph.

\begin{RULE}[Reducing block-cut vertex paths]\label{rule:blockcutpath}
Let $t_1t_2t_3t_4$ be an induced path of $G$ and
for each $1\le i\le 3$, let $S_i\subseteq V(G)\setminus \{t_1, \ldots, t_4\}$ be a clique of $G$ such that
\begin{itemize}
\item for each $1\le i\le 3$ and $v\in S_i$, $N_G(v)\setminus S_i=\{t_i, t_{i+1}\}$, and
\item for each $2\le i\le 3$, $N_G(t_{i})=\{t_{i-1}, t_{i+1}\}\cup S_{i-1}\cup S_{i}$. 
\end{itemize}
Then we remove $S_2$ and contract $t_2t_3$. 
\end{RULE}

Clearly, we can apply Reduction Rule~\ref{rule:blockcutpath} in polynomial time. We prove the soundness of Reduction Rule~\ref{rule:blockcutpath}.

\begin{figure}[t]
\centering
\begin{tikzpicture}[->,>=stealth',shorten >=1pt,auto,node distance=1.5cm,
                    semithick]
  \tikzstyle{big_graph}=[fill=black,ellipse, draw=none,text=white, minimum size=40pt]
  \tikzstyle{normal}=[draw,circle,fill=white,minimum size=4pt, inner sep=0pt]
  \tikzstyle{clique}=[draw,ellipse,fill=gray,minimum size=20pt, inner sep=0pt]
  \tikzstyle{ghost}=[draw=none,ellipse,fill=none,minimum size=4pt, inner sep=0pt]

  \node[big_graph] 	(A)                    {Rest of $G$};
  \node[ghost]		(B) [below of=A] {};
  \node[ghost]		(C) [below of=B] {};
  \node[clique]        	(D) [below  of=C] {$S_2$};
  \node[normal]         (E) [left of=B] {$t_1$};
  \node[clique]    	(F) [left of=C] {$S_1$};
  \node[normal]		(G) [left of=D] {$t_2$};
  \node[normal]		(H) [right of=B] {$t_4$};
  \node[clique]		(I) [right  of=C] {$S_3$};
  \node[normal]		(J) [right of=D] {$t_3$};

  \path (A) 	edge[-, line width=1pt] node {} (H)
  			edge[-, line width=1pt, bend left] node {} (H)
			edge[-, line width=1pt, bend right] node {} (H)
  			edge[-, line width=1pt] node {} (E)	 
  			edge[-, line width=1pt, bend left] node {} (E)	 
  			edge[-, line width=1pt, bend right] node {} (E)	 
            (F) 	edge[-, line width=4pt] node {} (E)
            		edge[-, line width=4pt] node {} (G)
		(D) 	edge[-, line width=4pt] node {} (G)
			edge[-, line width=4pt] node {} (J)
		(I)	edge[-, line width=4pt] node {} (H)
			edge[-, line width=4pt] node {} (J) ;
\end{tikzpicture} 
\caption{Reduction Rule~\ref{rule:blockcutpath}} \label{blockcutpath}
\end{figure}
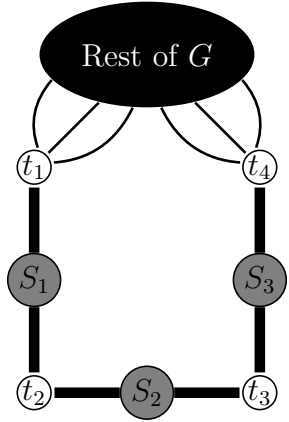

\begin{LEM}\label{lem:soundrule5}
Reduction Rule~\ref{rule:blockcutpath} is safe.
\end{LEM}
\begin{proof}
Let $t_1t_2t_3t_4$ be an induced path of length 3 in $G$ and
for each $1\le i\le 3$, $S_i$ be a clique which altogether satisfy the condition of Reduction Rule~\ref{rule:blockcutpath}.

It is easy to check that no vertex from $S_2$ is contained in an induced cycle of length at least 4, or an induced diamond in $G$. Since all obstructions are 2-connected, any obstruction in $G$ intersecting $S_2\cup \{t_2,t_3\}$ contains exactly $t_2,t_3$ and none of $S_2$. This means that such an obstruction is a cycle of length at least 5, which remains an obstruction after deleting $S_2$ and contracting the edge $t_2t_3$. Thus, $(G, k)$ is a \YES-instance if and only if $(G', k)$ is a \YES-instance.
\end{proof}

The following rule will be applied using Proposition~\ref{prop:generalcompletedegree}. 

\begin{RULE}[$(k+1)$-distinct obstructions rule]\label{rule:disjoint}
Let $v\in V(G)$ and let $G':=G-v-E(G[N_G(v)])$ such that there are $2k+1$ vertex-disjoint $N_G(v)$-paths in $G'$. 
If $G$ contains $k+1$ vertex-disjoint obstructions, then say that it is a \NO-instance.
Otherwise, we remove $v$ from $G$, and decrease $k$ by one. (By Proposition~\ref{prop:generalcompletedegree}, one of them exists.)
\end{RULE}

\subsection{A vertex of large complete degree}

An instance $(G, k)$ is called a \emph{reduced instance} if it is reduced under Rules~\ref{rule:blockcomponent}, \ref{rule:cutvertex}, \ref{rule:twinreduction}, \ref{rule:blockcutpath}, and \ref{rule:disjoint} introduced in the previous subsection. In this subsection, we prove that there exists a vertex of large complete degree whenever a reduced instance is sufficiently large, which is stated as Theorem~\ref{thm:largecompletedegree}.

For positive integers $k,\ell$, we define that 
\begin{itemize}
\item $g_1(k,\ell):=6k^2(\ell+14k)^2+2k(\ell+14k)$,
\item $g_2(k,\ell):=(k+1)^2+7k^2+\frac{1}{2}k(\ell+14k)$.
\end{itemize}

\begin{THM}\label{thm:largecompletedegree}
Let $(G, k)$ be a reduced instance of \BLOC\ that is a \YES-instance.
If $G$ has at least $k+g_1(k,\ell)g_2(k,\ell)$ vertices then $G$ has a vertex of complete degree  at least $\ell+1$.
\end{THM}

Let $(G, k)$ be a reduced instance of \BLOC\ and let $S\subseteq V(G)$ of size at most $k$ such that $G- S$ is a block graph. 
We let $G':=G- S$ and for each $v\in S$, we define that
\begin{itemize}
\item $G_v:=G[V(G')\cup \{v\}]$, 
\item $S_v'$ is a vertex set of size at most $7k$ in $G-v$ that is obtained by Proposition~\ref{prop:generalcompletedegree},
\item $S_v:=S_v'\cap V(G')$.
\end{itemize}
Let $T:=\bigcup_{v\in S} S_v$. Note that $\abs{T}\le 7k^2$ and for each $v\in S$, there are no block graph obstructions containing $v$ in $G_v- T$.

We first give a bound on the size of each block of $G'$ and the number of blocks in $G'$ sharing a cut vertex with it, assuming that there is no vertex in $S$ of large complete degree in $G$. 
Each block of $G'$ consists of the set of simplicial vertices and the set of cut vertices in $G'$. 

\begin{LEM}\label{lem:ss2}
Let $F$ be a graph whose vertex set is $X\cup \{v_1, \ldots, v_{t}\}$ such that
$t\ge 2$ and $X$ is a clique of $F$ and every two vertices of $X$ have different neighbors on $\{v_1, \ldots, v_{t}\}$.
If $\abs{X}\ge t+2$, then $F$ contains a diamond having exactly one vertex of $\{v_1, \ldots, v_{t}\}$.
\end{LEM}
\begin{proof}
Without loss of generality, we can assume that $\{v_1, \ldots, v_{t}\}$ is a minimal set with the aforementioned property. Notice that there exists a vertex $v_i$ which has at least two neighbors in $X$. By minimality assumption, $v_i$ is not adjacent with all vertices in $X$. Choose distinct vertices $x,y,z\in X$ such that $x$,$y$ are neighbors of $v_i$ and $z$ is not. Observe that $F[\{v_i,x,y,z\}]$ is isomorphic to the diamond containing exactly one vertex of $\{v_1, \ldots, v_{t}\}$
\end{proof}

\begin{LEM}\label{lem:boundunmarked} 
Let $B$ be a block of $G'$, and let $B_1$ and $B_2$ be the sets of all simplicial vertices and all cut vertices of $G'$ contained in $B$, respectively.
Let $H_1, H_2, \ldots, H_{t}$ be the components of $G'- V(B)$ that has a neighbor in $B$. The followings hold. 
\begin{enumerate}
\item $\abs{B_1}\le (k+1)^2+7k^2$.
\item 
If 
for every $v\in S$, $v$ has complete degree at most $\ell$ in $G$, then $\abs{B_2}\le t\le k(\ell+14k)$.
\end{enumerate}
\end{LEM}
\begin{proof}
(1) We first give a bound on the number of simplicial vertices of a block for $G'-T$.
Note that $B\setminus T$ is a block of $G'-T$.
Let $B_1'$ be $B_1\setminus T$. Clearly, $\abs{B_1}\le \abs{B_1'}+7k^2$.

Since vertices in $B_1'$ are pairwise true twins in $G'$, 
if two vertices in $B_1'$ have the same neighbors on $S$, then they are true twins in $G$.
We partition $B_1'$ into equivalent classes where two vertices are equivalent if they have the same neighbors on $S$.
From Reduction Rule~\ref{rule:twinreduction}, each equivalent class  has at most $k+1$ vertices.

If $\abs{S}\le 1$, then there are at most $2$ equivalent classes in $B_1'$.
If $\abs{S}\ge 2$ and the number of equivalent classes in $B_1'$ is at least $k+2$, 
then since $\abs{S}\le k$,
$G$ contains a diamond containing exactly one vertex $v$ of $S$ by Lemma~\ref{lem:ss2}.
This contradicts to the fact that $G_v-T$ has no obstruction containing $v$.
Thus, the number of equivalent classes in $B_1'$ is at most $k+1$ and $\abs{B_1}\le \abs{B_1'}+7k^2\le (k+1)^2+7k^2$.

\vskip 0.2cm
(2) Suppose that for every $v\in S$, $v$ has complete degree at most $\ell$ in $G$.
It means that there is a way to remove $7k$ vertices from the neighborhood of $v$ in $G$
so that the number of the remaining components is at most $\ell$.
Since removing the set $S_v'$ also makes the neighborhood of $v$ into a disjoint union of complete graphs, by Lemma~\ref{lem:countingcomponents}, $G[N_{G}(v)\setminus S_v']$ has at most $\ell+7k$ components.
In particular, $G'[N_{G_v}(v)\setminus S_v]=G[N_{G}(v)\setminus S_v']-S$ also has at most $\ell+7k$ components that are complete graphs 
 as the number of components cannot increase when removing vertices.

From Reduction Rule~\ref{rule:cutvertex}, we may assume that each $H_i$ contains at least one neighbor of a vertex in $S$.
On the other hand, each $v$ in $S$ has at most $\ell+7k$ complete neighbors except $7k$ neighbors in $S_v$ and one complete neighbor cannot belong to two components of $H_1, \ldots, H_{t}$.
Thus, if $t\ge k(\ell+14k)+1$,
then there exists $H_j$ for some $1\le j\le t$ such that there are no edges between $S$ and $V(H_j)$, which is contradiction.
Since each component of $H_1, \ldots, H_t$ has at most one neighbor in $B_2$, we conclude that $\abs{B_2}\le t\le k(\ell+14k)$.
\end{proof}

\noindent {\bf Contracted Block Tree.} We introduce a notion called the {\em contracted block tree} of $G$. A contracted block tree ${\cal T}_G$ of a connected graph $G$ is a rooted tree obtained from a block tree ${\cal T}^0_G$ of $G$ by (i) choosing a block vertex of ${\cal T}^0_G$ as a root, and (ii) for each cut vertex $c$ of ${\cal T}^0_G$, identifying it with its unique parent.

Let ${\cal T}_{G'}$ be the union of contracted block trees of connected components of $G'$. 
We color the vertices of ${\cal T}_{G'}$ in three phases: in the first phase, for every vertex $v\in S$ and for every $w\in N_{V(G)\setminus S}(v)$, we choose the (unique) block $B\in V({\cal T}_{G'})$ which contains $w$ and is closest to the root, and color $B$ by red. Let $R_1$ be the vertices colored red so far. In the second phase, we again recursively color the least common ancestor of any pair of red vertices by red. Let $R$ be the set of red vertices ${\cal T}_{G'}$. 
All other vertices of ${\cal T}_{G'}$ are colored blue.
 
\begin{LEM}\label{lem:countingred}
Suppose that the complete degree of $v$ is at most $\ell$ for every $v\in S$. Then we have $\abs{R}\leq 2k(\ell+14k)$.
\end{LEM}

\begin{proof}
It is easy to see that $\abs{R}\leq 2\abs{R_1}-1$, so we prove that $\abs{R_1}\leq k(\ell+14k)$. For each $v\in S$, the neighborhood $N_G(v)\setminus S$ can be partitioned into two sets: those contained in $S_v$ and $N_{G_v}(v)\setminus S_v$. Recall that $G[N_{G_v}(v)\setminus S_v]$ is a disjoint union of complete graphs, and there are at most $\ell+7k$ of them since the complete degree of $v$ is at most $\ell$ and due to Lemma~\ref{lem:countingcomponents}. Each complete graph in $G[N_{G_v}(v)\setminus S_v]$ is entirely contained in a block of $G'$, and thus renders at most one block vertex of ${\cal T}_{G'}$ red. With $\abs{S_v}\leq 7k$, it follows that for each $v\in S$, at most $\ell+14k$ block vertices are colored red in the first phase. Hence, we have $\abs{R_1}\leq k(\ell+14k)$.
\end{proof}

\begin{LEM}\label{lem:bluecomponents}
Let $T$ be a tree with at least $2$ vertices and degree at most $d$, and let $M$ be a set of vertices in $T$. Then there are at most $d\cdot \abs{M}$ connected components in $T-M$. 
\end{LEM}
\begin{proof}
We use induction on $\abs{M}$. If $\abs{M}=1$, then it is clear and we assume that $\abs{M}\ge 2$.
Let $r\in V(T)$ be the root of $T$ and orient all edges of $T$ toward $r$. Choose a vertex $v\in M$ farthest from the root and let $T_v$ be the subtree rooted at $v$. By induction hypothesis, the number of connected components in $T-V(T_v)-(M\setminus \{v\})$ is at most $d\cdot (\abs{M}-1)$. Therefore, the number of connected components in $T-M$ is at most $d\cdot (\abs{M}-1) + (d-1) \leq d \cdot \abs{M}$ as claimed.
\end{proof}

The next lemma follows from Lemma~\ref{lem:countingred} and~\ref{lem:bluecomponents}.
\begin{LEM}\label{lem:longbluepath} 
If $G'$ contains at least $g_1(k,\ell)$ blocks,  
then  ${\cal T}_{G'}$ has a blue component on at least $3$ vertices. 
\end{LEM}
\begin{proof}
Notice that every component in ${\cal T}_{G'}$ has at least one red vertex since $(G,k)$ is reduced with respect to Reduction Rule~\ref{rule:blockcomponent}. The degree of ${\cal T}_{G'}$ is bounded by $k(\ell+14k)$ by Lemma~\ref{lem:boundunmarked}. By combining Lemmata~\ref{lem:countingred} and~\ref{lem:bluecomponents} with $d=k(\ell+14k)$, we observe that the number of blue components of ${\cal T}_{G'}-R$ is bounded by 
\[d\cdot \abs{R}\leq k(\ell+14k)\cdot 2k(\ell+14k)=2k^2(\ell+14k)^2\]
Also, the total number of blue vertices in ${\cal T}_{G'}$ is at least
\[\abs{V({\cal T}_{G'})}-\abs{R} \geq\abs{V({\cal T}_{G'})}- 2k(\ell+14k)\ge 6k^2(\ell+14k)^2,\]
and therefore, ${\cal T}_{G'}$ has a blue component having at least $3$ vertices. 
\end{proof}

Lemma~\ref{lem:longbluepath} and the property of two phase coloring is essential for the proof of our main result in this subsection.

\begin{proof}[Proof of Theorem~\ref{thm:largecompletedegree}]
Let $(G,k)$ be a reduced instance with $\abs{V(G)}\geq k+g_1(k,\ell)g_2(k,\ell)$ 
and $S\subseteq V(G)$ be a set of size at most $k$ such that $G-S$ is a block graph. To derive contradiction, suppose that for every $v\in S$, $v$ has complete degree at most $\ell$ in $G$.
Then $G'=G-S$ has at least $g_1(k,\ell)g_2(k,\ell)$
vertices.
Let $p$ be the number of blocks of $G'$. 
From Lemma~\ref{lem:boundunmarked} and the fact that each cut vertex is contained in at least two blocks, 
we obtain that
\[  \abs{V(G')}\le p((k+1)^2+7k^2) + \frac{1}{2}pk(\ell+14k) \le p\cdot g_2(k,\ell).\]
Therefore, we have $p\ge g_1(k,\ell)$.
By Lemma~\ref{lem:longbluepath}, 
$\mathcal{T}_{G'}$ contains a blue component $P$ on at least $3$ vertices. 

We claim that $P$ is (i) a path, and (ii) each of its two end vertices, and no other, is adjacent with exactly one red vertex. Let us prove (i) first. Let $W$ be the unique block vertex in $P$ which is closest to the root. Notice that $W$ is not the root itself since the instance is reduced with respect to Reduction Rule~\ref{rule:blockcomponent} and thus the root is a red vertex. Hence $W$ has a unique parent which is red. 
For any $Z$ which is a leaf in the subtree $P$, it is adjacent with at least one red vertex. Indeed, if not, $Z$ is a leaf in ${\cal T}_{G'}$. Then by Reduction Rule~\ref{rule:cutvertex}, the block $Z$ (possibly except for its unique cut vertex) should have been removed from $G$, a contradiction. Note that any red vertex adjacent with $Z$ is a child of $Z$ since the path from $Z$ to $W$ is blue and $W\neq Z$. Furthermore, the subtree $P$ has exactly one leaf since otherwise, the second phase of coloring must have colored the branching vertices contained in $P$, a contradiction. This establishes (i). For (ii), observe that if (ii) does not hold, then some vertex of $P$ must have been colored in the second phase, a contradiction.

Now, with $P$ together with the two red vertices incident with $V(P)$, we can apply Reduction Rule~\ref{rule:blockcutpath}, a contradiction. Therefore, we conclude that there exists a vertex $v\in S$ such that $v$ has complete degree at least $\ell+1$ in $G$.
\end{proof}

\section{Reducing the instance with large complete degree}\label{sec:largecomplete}

We introduce the last rule, which will be used when $G$ has a vertex of large complete degree.
We use the well-known technique, called the $\alpha$-expansion lemma, which is already used in several kernelization algorithms~\cite{Thomasse2009, CyganPPW2010, MisraPRS2012, CyganPPO2013}.
One notable difference from other approaches is that, to guarantee the equivalence, we add some paths in the given graph, and thus increase the number of vertices. However, we show that our rule decreases $n+m^*$ where $m^*$ is the number of edges whose both degrees are at least $3$, by using the $3$-expansion lemma instead of the $2$-expansion lemma.

\begin{RULE}[Large complete degree rule]\label{rule:largecompletedegree} 
Let $v\in V(G)$ and $X\subseteq V(G)\setminus \{v\}$ with $\abs{X}\le 7k$.
Let $\mathcal{C}$ be a set of connected components of $G-(X\cup \{v\})$ and let $\phi:X\rightarrow \binom{\mathcal{C}}{3}$ such that
\begin{itemize}
\item for each $C\in \mathcal{C}$, $G[\{v\}\cup V(C)]$ is a block graph, $v$ has a neighbor in $C$, and there exists a vertex $x\in X$ that has a neighbor in $C$,  
\item for $x\in X$, $\phi(x)$ is a subset of $\mathcal{C}$ where each graph in $\phi(x)$ has a neighbor of $x$, and  
\item the sets in $\{\phi(x) :x\in X\}$ are pairwise disjoint.
\end{itemize}
Then, remove all edges between $v$ and every component of ${\cal C}$, and add two internally vertex-disjoint paths of length two between $v$ and each vertex $x\in X$. (All of the new vertices in these paths have degree $2$ in the resulting graph). If a component of  $\mathcal{C}$ has a vertex of degree $1$ in the resulting graph, then we remove the vertex. See Figure~\ref{fig:rule6}.
\end{RULE}

\begin{figure}\center
  \tikzstyle{v}=[circle, draw, solid, fill=black, inner sep=0pt, minimum width=3pt]
  \begin{tikzpicture}[scale=0.06,rotate=-90]
    \draw (0,30) node [above right] {$v$} {};
    \draw (0,30) node [v] {}; 

    \draw[|<->|,thick] (40,50)
    -- node[below] {size $\le 7k$} (40,10);
    
     \draw [color=gray]  (30,30) ellipse (4 and 20);
    
    \foreach \x in {10,20,..., 50}
    {
      \draw (0,30) -- (10,\x)node[v]{} [snake=bumps]--(20,\x) node [v]{};
      \draw [color=gray]  (15,\x) ellipse (7 and 3);
    }
      \draw [color=gray, fill=yellow]  (15,-10) ellipse (7 and 3);
      \draw [color=gray, fill=yellow]  (15,0) ellipse (7 and 3);
      \draw [color=gray, fill=yellow]  (15,10) ellipse (7 and 3);
      \draw [color=gray, fill=yellow]  (15,20) ellipse (7 and 3);
      \draw [color=gray, fill=yellow]  (15,30) ellipse (7 and 3);
      \draw [color=gray, fill=yellow]  (15,40) ellipse (7 and 3);
      \draw (0,30) -- (10,-10)node[v]{} [snake=bumps]--(20,-10) node [v]{};
      \draw (0,30) -- (10,0)node[v]{} [snake=bumps]--(20,0) node [v]{};
      \draw (0,30) -- (10,10)node[v]{} [snake=bumps]--(20,10) node [v]{};
      \draw (0,30) -- (10,20)node[v]{} [snake=bumps]--(20,20) node [v]{};
      \draw (0,30) -- (10,30)node[v]{} [snake=bumps]--(20,30) node [v]{};
      \draw (0,30) -- (10,40)node[v]{} [snake=bumps]--(20,40) node [v]{};
\draw [blue, fill=yellow] (25,10) rectangle (35,22);    
    
   \draw (30,12)node[v]{} --(20,0) node [v]{};
   \draw (30,12)node[v]{} --(20,-10) node [v]{};

   \draw (30,12)node[v]{} --(20,10) node [v]{};
   \draw (30,20)node[v]{} --(20,10) node [v]{};
   \draw (30,12)node[v]{} --(20,20) node [v]{};
   \draw (30,20)node[v]{} --(20,20) node [v]{};
   \draw (30,12)node[v]{} --(20,30) node [v]{};
   \draw (30,20)node[v]{} --(20,40) node [v]{};
  
   \draw (30,32)node[v]{} --(20,50) node [v]{};
   \draw (30,38)node[v]{} --(20,50) node [v]{};
   \draw (30,42)node[v]{} --(20,50) node [v]{};

  \end{tikzpicture}\qquad\quad
    \begin{tikzpicture}[scale=0.06,rotate=-90]
    \draw (0,30) node [above right] {$v$} {};
    \draw (0,30) node [v] {}; 

    \draw[|<->|,thick] (40,50)
    -- node[below] {size $\le 7k$} (40,10);
    
     \draw [color=gray]  (30,30) ellipse (4 and 20);
    
    \foreach \x in {10,20,..., 50}
    {
      \draw (10,\x)node[v]{} [snake=bumps]--(20,\x) node [v]{};
      \draw [color=gray]  (15,\x) ellipse (7 and 3);
    }
       \draw [color=gray, fill=yellow]  (15,-10) ellipse (7 and 3);
      \draw [color=gray, fill=yellow]  (15,0) ellipse (7 and 3);
     \draw [color=gray, fill=yellow]  (15,10) ellipse (7 and 3);
      \draw [color=gray, fill=yellow]  (15,20) ellipse (7 and 3);
      \draw [color=gray, fill=yellow]  (15,30) ellipse (7 and 3);
      \draw [color=gray, fill=yellow]  (15,40) ellipse (7 and 3);
      \draw (10,10)node[v]{} [snake=bumps]--(20,10) node [v]{};
      \draw (10,20)node[v]{} [snake=bumps]--(20,20) node [v]{};
    \draw (10,30)node[v]{} [snake=bumps]--(20,30) node [v]{};
      \draw (10,40)node[v]{} [snake=bumps]--(20,40) node [v]{};
  
\draw [blue, fill=yellow] (25,10) rectangle (35,22);    

     \draw (0,30) -- (10,50);
     \draw (30,12)node[v]{} --(20,0) node [v]{};
   \draw (30,12)node[v]{} --(20,-10) node [v]{};
   
    \draw (10,-10)node[v]{} [snake=bumps]--(20,-10) node [v]{};
      \draw (10,0)node[v]{} [snake=bumps]--(20,0) node [v]{};
      \draw (30,12)node[v]{} --(20,10) node [v]{};
   \draw (30,20)node[v]{} --(20,10) node [v]{};
   \draw (30,12)node[v]{} --(20,20) node [v]{};
   \draw (30,20)node[v]{} --(20,20) node [v]{};
   \draw (30,12)node[v]{} --(20,30) node [v]{};
   \draw (30,20)node[v]{} --(20,40) node [v]{};
  
   \draw (30,32)node[v]{} --(20,50) node [v]{};
   \draw (30,38)node[v]{} --(20,50) node [v]{};
   \draw (30,42)node[v]{} --(20,50) node [v]{};
   

         \draw (0,30) -- (0,10)-- (30,12);
    \draw (0,30) --(2, 12)-- (30, 12);
         \draw (0,30) -- (4,14)-- (30,20);
    \draw (0,30) -- (6, 16)--(30, 20);

  \end{tikzpicture}
  \caption{Reduction Rule 6}
  \label{fig:rule6}
\end{figure}
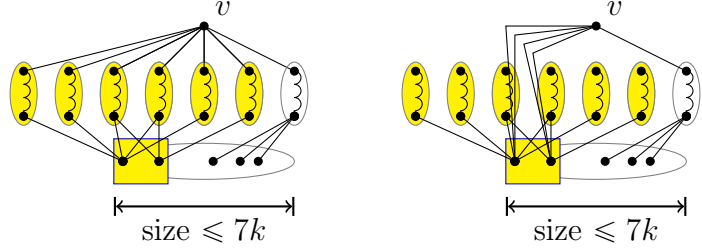

As we discussed, we clarify that it decreases $n+m^*$ where $m^*$ is the number of edges whose both end vertices have degree at least $3$.
Since $\abs{\mathcal{C}}\ge 3\abs{X}$ and $n+m^*$ is increased by $2\abs{X}$ by adding paths of length $2$ from $v$ to each vertex of $X$,
it is sufficient to show that for each $C\in \mathcal{C}$, $n+m^*$ is decreased by at least $1$ by removing the edges between $v$ and $C$.
Let $C\in \mathcal{C}$. If $\abs{N_G(v)\cap C}\ge 3$, then it is trivial.
First assume that $\abs{N_G(v)\cap C}=2$. Then $C$ has more than two vertices, or there exists a vertex $x\in X$ that has a neighbor on $N_G(v)\cap C$. In either case, it is not difficult to verify that one of the vertex in $N_G(v)\cap C$ has degree at least $3$ in $G$.
Therefore, $m^*$ is decreased by at least $1$ when removing the edges between $v$ and $C$.
Now, let us assume that $N_G(v)\cap C=\{w\}$ for some $w\in V(C)$.
If $w$ has degree $2$, then after removing the edge $vw$, we also remove $w$ following Reduction Rule~\ref{rule:largecompletedegree}.
Thus, $n$ is decreased by $1$. Otherwise, removing $vw$ decreases $m^*$ by $1$.
We conclude that $n+m^*$ is always decreased when applying Reduction Rule~\ref{rule:largecompletedegree}.

Now we describe how to obtain a polynomial-size kernel from a given instance. The algorithm presented in the following theorem is used as a subroutine.
\begin{THM}[$\alpha$-expansion lemma~\cite{Thomasse2009}]\label{lem:expansionlemma}
Let $\alpha$ be a positive integer. Let $F$ be a bipartite graph on the bipartition $(X, Y )$ with $\abs{Y}\ge \alpha \abs{X}$ such that every vertex of $Y$ has at least one neighbor in $X$. Then there exist nonempty subsets $X'\subseteq X$ and $Y'\subseteq Y$ and a function $\phi:X'\rightarrow \binom{Y'}{\alpha}$ such that 
\begin{itemize}
\item $N_F(Y')\cap X=X'$,
\item $\phi(x)\subseteq N_F(x)$ for each $x\in X'$, and  
\item the sets in $\{\phi(x):x\in X'\}$ are pairwise disjoint.
\end{itemize}
In addition, such pair of subsets $X', Y'$ can be computed in polynomial time in $\alpha\abs{V(F)}$.
\end{THM}

\begin{THM}\label{thm:reducinglargecomplete}
Reduction Rule~\ref{rule:largecompletedegree} is safe.
\end{THM}
\begin{proof}
Let $G$ be a graph and let $v\in V(G)$ and $X\subseteq V(G)\setminus \{v\}$ with $\abs{X}\le 7k$.
Let $\mathcal{C}$ be a set of connected components of $G-(X\cup \{v\})$ and let $\phi:X\rightarrow \binom{\mathcal{C}}{3}$ such that
\begin{itemize}
\item for each $C\in \mathcal{C}$, $G[\{v\}\cup V(C)]$ is a block graph, 
\item $\phi(x)$ is a subset of $\mathcal{C}$ whose components have a neighbor of $x$, and  
\item the graphs in $\{\bigcup_{C\in \phi(x)}V(C) :x\in X\}$ are pairwise disjoint.
\end{itemize}
Let $G'$ be the resulting graph obtained by using Reduction Rule~\ref{rule:largecompletedegree}, and
let $R$ be the new vertices of degree $2$ linking between $v$ and $X$ in $G'$.
We prove that $(G,k)$ is a \YES-instance if and only if $(G', k)$ is a \YES-instance.

First suppose that $G'$ has a vertex set $A$ with $\abs{A}\le k$ such that $G'-A$ is a block graph.
Suppose a vertex $r\in R$ is contained in $A$ and let $r'$ be a neighbor of $r$.
Then $G'-(A\setminus \{r\}\cup \{r'\})$ is also a block graph, as $r$ and the twin of $r$ become vertices of degree $1$ in $G'-(A\setminus \{r\}\cup \{r'\})$ and thus they cannot be contained in any obstruction.
Since any two paths of length $2$ traversing $R$ form an induced subgraph isomorphic to $C_4$, 
we may assume that $A$ contains one of the neighbors of $r$.
That is,  
we have $v\in A$ or $X\subseteq A$.
If $v\in A$, then $G-A$ is an induced subgraph of $G'-A$, and therefore, $G-A$ is a block graph.
Suppose $X\subseteq A$ and let $B$ be a obstruction in $G-A$.
Then $B$ cannot be contained in $G[(\bigcup_{C\in \mathcal{C}} V(C)) \cup \{v\}]$ because 
$G[(\bigcup_{C\in \mathcal{C}} V(C)) \cup \{v\}]$ is a block graph.
Thus, $B$ should be contained in $G-A-(\bigcup_{C\in \mathcal{C}} V(C))$ that is an induced subgraph of $G'-A$, and it contradicts to that $G'-A$ is a block graph.

Now suppose that $G$ has a vertex set $A$ with $\abs{A}\le k$ such that $G-A$ is a block graph.
If $v\in A$, then it is easy to observe that $G'-A$ is a block graph as degree $1$ vertices cannot be contained in an obstruction.
Hence, we may assume that $v\notin A$.
Let $A_1:=X\setminus A$ and $A_2:=A\cap (\bigcup_{C\in \mathcal{C}}V(C))$.
It is not hard to see that $G-(A\setminus A_2\cup A_1)$ is also a block graph as for each $C\in \mathcal{C}$, $G[\{v\}\cup V(C)]$ is a block graph and $N_G(C)\subseteq \{v\}\cup X$.
Now we check that $\abs{A_2}\ge \abs{A_1}$. For contradiction, suppose $\abs{A_2}<\abs{A_1}$.
Since the graphs in $\{\bigcup_{C\in \phi(x)}V(C) :x\in X\}$ are pairwise disjoint, 
there exists a vertex $a$ in $A_1$ such that $\phi(a)$ contains no vertex from $A_2$. 
Then two components in $\phi(a)$ with the vertices $v$ and $a$ forms an induced cycle of length at least $4$, which is contradiction.
Thus, $\abs{A_2}\ge \abs{A_1}$, and therefore $A\setminus A_2\cup A_1$ is also a proper deletion set of size at most $k$ in $G$.
As all vertices in $R$ become vertices of degree $1$ in $G'-(A\setminus A_2\cup A_1)$, 
$G'-(A\setminus A_2\cup A_1)$ is a block graph, as required.
\end{proof}

\begin{proof}[Proof of Theorem~\ref{thm:main1}]
Given an instance $(G,k)$, we exhaustively apply Reduction Rules \ref{rule:blockcomponent}-\ref{rule:disjoint} to obtain a reduced instance. 
If a reduced graph $G$ has at least $k+g_1(k,29k)g_2(k,29k)$ vertices, then 
by Theorem~\ref{thm:largecompletedegree}, $G$ has a vertex of complete degree  at least $29k$.
By Proposition~\ref{prop:generalcompletedegree}, 
we can find in polynomial time a vertex $v$ and a vertex set $S_v\subseteq V(G-v)$ such that $G-S_v$ has no block graph obstruction containing $v$, and 
$G[N_G(v)\setminus S_v]$ has at least $29k-7k=22k$ components.
Note that there are at most $k$ components of $G- (\{v\}\cup S_v)$ that may contain an obstruction, and
for each component $C$ of $G- (\{v\}\cup S_v)$, at most one component of $G[N_G(v)\setminus S_v]$ can be contained in $C$.
Let $\mathcal{C}$ be the set of components of $G- (\{v\}\cup S_v)$ which (i) contains a component of $G[N_G(v)\setminus S_v]$, and (ii) has no block graph obstructions.
Since $\abs{\mathcal{C}}\ge 22k-k=21k$ and $\abs{S_v}\le 7k$,  
using Theorem~\ref{lem:expansionlemma}, we can find in polynomial time sets $\mathcal{C}'\subseteq \mathcal{C}$ and $S_v'\subseteq S_v$ and a function $\phi:S_v'\rightarrow \binom{\mathcal{C}'}{3}$ such that
\begin{itemize}
\item the set of vertices in $S_v$ that has a neighbor in $\bigcup_{C\in \mathcal{C}'}V(C)$ is $S_v'$, 
\item for $x\in S_v'$, $\phi(x)$ is a subset of $\mathcal{C}$ where each graph in $\phi(x)$ has a neighbor of $x$, and  
\item the sets in $\{\bigcup_{C\in \phi(x)}V(C) :x\in S_v'\}$ are pairwise disjoint.
\end{itemize}
Note that for each $C\in \mathcal{C}'$, $G[\{v\}\cup V(C)]$ is a block graph, otherwise, it has an obstruction containing $v$, contradicting to the definition of $S_v$.
Furthermore, for each $C\in \mathcal{C}'$, there exists a vertex $x\in S_v'$ that has a neighbor in $C$, otherwise, we can reduce it using Reduction Rule~\ref{rule:cutvertex}.
So, we can apply Reduction Rule~\ref{rule:largecompletedegree} to reduce this instance. 
We apply these reductions recursively.
As we discussed, each step decreases $n+m^*$ where $m^*$ is the number of edges whose both end vertices have degree $3$, so, it will terminate in polynomial time, and at the final step, the resulting graph will have less than $k+g_1(k, 29k)g_2(k, 29k)=\mathcal{O}(k^{6})$ vertices.
\end{proof}
	
\section{A fixed parameter tractable algorithm}

The goal of this section is to prove Theorem~\ref{thm:main2} claiming an $O(10^k\cdot n^{O(1)})$-time algorithm for \BLOC. We apply iterative compression technique, which is established as a powerful tool to design FPT algorithms since it was first introduced by Reed, Smith and Vetta~\cite{ReedSV2004}. Our algorithm \BLOC\  requires as a subroutine an FPT algorithm for  the following disjoint version of \BLOC.

\vspace{2mm}

\noindent {\bf \disjointBLOC}
\begin{description}
\item[Input] A graph $G$, $S\subseteq V(G)$ such that both $G-S$ and $G[S]$ are block graphs, an integer $k$.
\item[Parameter] $k$.
\item[Task] Find a {\em solution} to $(G,S,k)$, i.e. a set $\tilde{S}\subseteq V(G)\setminus S$ such that $G-\tilde{S}$ is a block graph and $\abs{\tilde{S}}\leq k$, or correctly report that no such set exists.
\end{description}

We present an algorithm {\bf Block}$(G,S,k)$ which solves \disjointBLOC\ in time $O(3^{k+\ell}\cdot n^6)$, where $\ell$ is the number of connected components in $G[S]$. 

\begin{algorithm}
\caption{Algorithm for \BLOC}
\begin{algorithmic}[1]

\Procedure {{\bf Block}}{$G,S,k$}

\State {\bf if }$k \geq 0$ and $G$ is a block graph, {\bf return} $\emptyset$. \label{line:yes}
\State {\bf if} $k\leq 0$ and $V(G)\setminus S\neq \emptyset$, {\bf return} {\sc No}. \label{line:no}
\If {$u,v,w\in V(G)\setminus S$ such that $G[S\cup \{u,v,w\}]$ is not a block graph} \label{line:ifsmallobs}
	\State \Comment{$u,v,w$ are not necessarily distinct if $\abs{V(G)\setminus S}\leq 2$}
	\State {\bf Block}$(G-u,S,k-1)\cup \{u\}$	\Comment{Small Set Branching Rule}			\label{line:smallobs1}
	\State {\bf Block}$(G-v,S,k-1)\cup \{v\}$				\label{line:smallobs2}
	\State {\bf Block}$(G-w,S,k-1)\cup \{w\}$				\label{line:smallobs3}
\ElsIf {there is $uv\in E(G-S)$ and $x,y\in N_S(\{u,v\})$ such that\\			\label{line:ifcomponent}
\qquad \quad $x,y$ belong to distinct connected components of $G[S]$}
		\State  {\bf Block}$(G-u,S,k-1)\cup \{u\}$	\Comment{Component Branching Rule}	\label{line:component1}
		\State  {\bf Block}$(G-v,S,k-1)\cup \{v\}$											\label{line:component2}
		\State  {\bf Block}$(G,S\cup\{u,v\},k)$												\label{line:component3}
\Else
	\State Let $B$ be a leaf block of $G-S$ and $\partial_{G-S}(B)=\{b\}$.					\label{line:bypass1}
	\State $G'\leftarrow G- B\setminus \partial_{G-S}(B) +\{bw:w\in N_S(B)\} $	 \Comment{Bypass Rule} \label{line:bypass2}
	\State  {\bf Block}$(G',S,k)$.												\label{line:bypass3}			\label{line:bypass2}
\EndIf
\EndProcedure 
\end{algorithmic}
\end{algorithm}

Let us establish that {\bf Block}($G,S,k$) correctly returns a solution to $(G,S,k)$ if it is a {\sc Yes}-instance, and returns {\sc No} otherwise. Notice that if $(G,S,k)$ does not meet the condition at line~\ref{line:yes}, then $V(G)\setminus S$ is non-empty and thus one of the steps at lines~\ref{line:no},~\ref{line:ifsmallobs},~\ref{line:ifcomponent}, or~\ref{line:bypass1} will be executed and some output will be returned at the end of the algorithm {\bf Block}$(G,S,k)$. The execution of {\bf Block}($G,S,k$) can be represented by a search tree where each node corresponds to a call made during the execution. For the correctness of the algorithm, we use induction on the level of a call in the search tree. It is clear that 
lines~\ref{line:yes}--\ref{line:no}, corresponding to the base case, returns the output correctly. 
If the condition at line~\ref{line:ifsmallobs} is met, then any solution $\tilde{S}$ to $(G,S,k)$ must contain one of $u,v$ and $w$. Conversely, if $\tilde{S}$ is a solution returned by one of the calls {\bf Block} at lines~\ref{line:smallobs1}--\ref{line:smallobs3}, then $\tilde{S}$ together with $u,v,$ or $w$ is a solution to $(G,S,k)$. To see the correctness of lines~\ref{line:component1}--\ref{line:component3}, first notice that they enumerate all possible intersection of a solution $\tilde{S}\cap \{u,v\}$. Hence it suffices to verify that $G[S\cup \{u,v\}]$ is indeed a block graph. This is a consequence from the fact that $G$ does not meet the condition of line~\ref{line:ifsmallobs} for any (at most) three vertices.

The branching rules considered at lines~\ref{line:ifsmallobs}-\ref{line:smallobs3} and lines~\ref{line:ifcomponent}-\ref{line:component3} are called the \texttt{Small Set Branching} and \texttt{Component Branching}, respectively. Notice that an instance $(G,S,k)$ considered at line~\ref{line:bypass1} is {\em reduced with respect} to \texttt{Small Set Branching} and \texttt{Component Branching} or, simply put, {\em irreducible}: neither branching rules apply to $(G,S,k)$. For the correctness of the algorithm {\bf Block}, it remains to show that \texttt{Bypass Rule} at line~\ref{line:bypass2} is safe, that is, $\tilde{S}$ is a solution to the instance $(G',S,k)$ at line~\ref{line:bypass3} if and only if it is a solution to $(G,S,k)$. We need the following lemmata.

\begin{LEM}\label{lem:adjblock}
Let $(G,S,k)$ be an irreducible instance and $B$ be a leaf block of $G- S$. Then either $N_S(B)=\emptyset$ or there exists a single block $X$ of $G[S]$ such that $N_S(B)\subseteq X$. 
\end{LEM}
\begin{proof}
Suppose the contrary and choose $x,y \in N_S(B)$ and two blocks $X,Y$ of $G[S]$ such that $x\in X\setminus Y$ and $y\in Y\setminus X$. Let $u,v$ be (not necessarily distinct) vertices of $B$ having $x,y$ as neighbors. As $(G,S,k)$ is reduced with respect to \texttt{Component Branching}, both $X$ and $Y$ belong to a single component of $G[S]$. Let $P$ be an $x,y$-path in $G[S]$ and observe that $P\cup \{u,v\}$ forms a cycle with $xy\notin E(G)$. This implies that $G[S\cup \{u,v\}]$ is not a block graph, contradiction to the assumption that $(G,S,k)$ is reduced with respect to \texttt{Small Set Branching}. 
\end{proof}

\begin{LEM}\label{lem:formblock}
Let $(G,S,k)$ be an irreducible instance and $B$ be a leaf block of $G- S$. Then $G[S\cup B]$ is a block graph.
\end{LEM}
\begin{proof}
If  $N_S(B)= \emptyset$, then $G[S\cup B]$ is trivially a block graph. Therefore, we assume that $N_S(B)\neq \emptyset$.
Let $X$ be the block of $G[S]$ containing all vertices of $N_S(B)$, which exists by Lemma~\ref{lem:adjblock}. It suffices to show that $G[X\cup B]$ is a block graph. Suppose not and let $C\subseteq X\cup B$ be a vertex set which induces an obstruction for block graphs. Recall that $(G,S,k)$ is reduced with respect to \texttt{Small Set Branching}, and thus $C$ contains at least four vertices of $B$. This means that $G[C]$ is an induced cycle of length at least 5. However, the vertices of $C\cap B$ are pairwise adjacent, which is impossible. This completes the proof of our statement.
\end{proof}

\begin{LEM}\label{lem:represent}
Let $(G,S,k)$ be an irreducible instance and $B$ be a leaf block of $G- S$. Then there exists a vertex $u\in B$ such that $N_S(u)=N_S(B)$. 
\end{LEM}
\begin{proof}
If $\abs{N_S(B)}\leq 1$, the statement trivially holds. Assume that $\abs{N_S(B)}\geq 2$. Suppose the contrary, and choose $u\in B$ and $x,y\in N_S(B)$ such that $ux \in E(G)$ and $uy\notin E(G)$. Since $y\in N_S(B)$, there exists $v\in B$ such that $vy\in E(G)$. By Lemma~\ref{lem:adjblock}, the two vertices $x$ and $y$ belong to a single block of $G-S$, and thus are adjacent. Notice that $\{u,v,x,y\}$ induces either a diamond or a cycle of length 4. However, $G[S\cup B]$ is a block graph by Lemma~\ref{lem:formblock}, a contradiction.
\end{proof}

\begin{LEM}\label{lem:anothersol}
Let $(G,S,k)$ be an irreducible instance and $B$ be a leaf block of $G-S$. If there is a vertex set $\tilde{S}\subseteq V(G)\setminus S$ such that $G-\tilde{S}$ is a block graph, then there is $\tilde{S}'\subseteq V(G)\setminus S$ such that $G-\tilde{S}'$ is a block graph, $\abs{\tilde{S}'}\leq \abs{\tilde{S}}$ and $\tilde{S}'\cap (B\setminus \partial_{G-S}(B))=\emptyset$.
\end{LEM}
\begin{proof}
Consider a vertex set $\tilde{S}\subseteq V(G)\setminus S$ such that $G-\tilde{S}$ is a block graph. If $\tilde{S}\cap (B\setminus \partial_{G-S}(B))=\emptyset$, then the statement trivially holds. Hence, suppose $\tilde{S}\cap (B\setminus \partial_{G-S}(B))\neq \emptyset$ and let $\tilde{S}'=(\tilde{S} \setminus B)\cup \partial_{G-S}(B)$. We want to show that $\tilde{S}'$ is a vertex set claimed by the statement.

Clearly, we have $\tilde{S}'\cap (B\setminus \partial_{G-S}(B))=\emptyset$. As $B$ is a leaf block in $G-S$, we have $\abs{\partial_{G-S}(B)}\leq 1$, which implies $\abs{\tilde{S}'}\leq \abs{\tilde{S}}$. To see that $G-\tilde{S}'$ is a block graph, suppose the contrary and let $C$ be a vertex set of $G-\tilde{S}'$ which induces an obstruction. Since $G-\tilde{S}$ is a block graph, any obstruction $C$ in $G-\tilde{S}'$ must contain some vertex $u$ of $B\setminus \tilde{S}'=B\setminus \partial_{G-S}(B)$. Moreover, $C$ contains some vertex $v\notin B\cup S$ since $G[B\cup S]$ is a block graph by Lemma~\ref{lem:formblock}. Let $X$ be a block such that $N_S(B)\subseteq X$, which exists by Lemma~\ref{lem:adjblock}. Notice that $C$ is 2-connected and $X$ is a separator between $u$ and $v$ in $G-\tilde{S}'$. This implies that $C$ also contains at least two vertices of $X$. Then, the obstruction $C$ cannot be an induced cycle and thus is a diamond. This means that $G[X\cup\{u,v\}]$, thus $G[S\cup\{u,v\}]$, is not a block graph, contradicting to the assumption that $(G,S,k)$ is reduced with respect to \texttt{Small Set Branching}. This proves that $G-\tilde{S}'$ is a block graph.
\end{proof}

The following lemma states the correctness of \texttt{Bypass Rule} applied at lines~\ref{line:bypass1}-\ref{line:bypass3}.
\begin{LEM}
Let $(G,S,k)$ be an irreducible instance, $B$ be a leaf block of $G- S$, and $G'$ be the graph obtained by applying \texttt{Bypass Rule}. 
\begin{itemize}
\item If $\tilde{S}$ is a solution to $(G,S,k)$, then $\tilde{S}\setminus (B\setminus \partial_{G-S}(B))$ is a solution to $(G',S,k)$.
\item If $\tilde{S}'$ is a solution to $(G',S,k)$, then it is also a solution to $(G,S,k)$.
\end{itemize}
\end{LEM}
\begin{proof}
Let $b$ be the unique cut vertex of $G-S$ contained in $B$. Let us prove the first implication. Suppose that $\tilde{S}$ is a solution to $(G,S,k)$ such that $\tilde{S}\cap (B\setminus \partial_{G-S}(B))=\emptyset$. Such a solution exists by Lemma~\ref{lem:anothersol}. We show that $\tilde{S}$ is a solution to $(G',S,k)$, from which the first implication follows. If $b\in \tilde{S}$, then $G'-\tilde{S}$ is clearly a block graph as it is an induced subgraph of $G-\tilde{S}$. Let us consider the case when $b\notin \tilde{S}$. For the sake of contradiction, suppose that $G'-\tilde{S}$ contains a vertex set $C$ inducing an obstruction. 
Consider a vertex $u\in B$ such that $N_S(u)=N_S(B)$. The existence of such $u$ is shown in Lemma~\ref{lem:represent}. Note that $u\neq b$ and there exists $x\in N_S(B)$ such that $bx\notin E(G)$ and $bx$ is contained in $C$, otherwise, $C$ also appears in $G-\tilde{S}$. 
If $C$ contains one more vertex from $N_S(B)$, then $C$ should be a diamond with two intersections on $N_S(B)$ in $G'-\tilde{S}$. Then $G[V(C)\setminus \{b\}\cup \{u\}]$ is a diamond of $G-\tilde{S}$, which is a contradiction. Thus, $\abs{V(C)\cap N_S(B)}=1$ and $G[V(C)\cup \{u\}]$ induces a subgraph isomorphic to a graph obtained from $C$ by subdividing one edge. It contains an obstruction in $G-\tilde{S}$, which contradicts to our assumption.

We establish the second implication. Suppose that $\tilde{S}'$ is a solution to $(G',S,k)$, but $G-\tilde{S}'$ is not a block graph. Let $C$ be a vertex set inducing an obstruction in $G-\tilde{S}'$. Then $G[C]$ is not a diamond nor a cycle of length 4 since otherwise, $G[C\cup S]$ is not a block graph and $\abs{C\setminus S}\leq 3$, contradicting to the assumption that $(G,S,k)$ is reduced with respect to \texttt{Small Set Branching}. Therefore $G[C]$ must be an induced cycle of length at least 5. Notice that $C$ contains some vertex $v\notin B\cup S$ since $G[B\cup S]$ is a block graph by Lemma~\ref{lem:formblock}. There are two possibilities, and in each case we derive a contradiction.

\medskip
\subparagraph{When $b \notin C$:} Notice that $N_S(B)\cap C$ is a separator between $B\cap C$ and $v$ in $G[C]$, and thus contains a minimal separator between $B\cap C$ and $v$. However, $N_S(B)$ is a complete graph by Lemma~\ref{lem:adjblock} while any minimal separator in an induced cycle must be non-adjacent, a contradiction.

\medskip
\subparagraph{When $b\in C$:} Observe that there is a vertex $x\in N_S(B)\cap C$ such that $x$ is adjacent with some vertex, say $w$, in $B\cap C$. We claim that $N_S(B)\cap C=\{x\}$. Suppose not, and let $y$ be a vertex in $(N_S(B)\cap C)\setminus \{x\}$. The existence of $v\in C\setminus (B\cup S)$ implies $wy\notin E(G)$. Take $u\in B$ such that $N_S(u)=N_S(B)$, which is possible due to Lemma~\ref{lem:represent}, and observe that $ux,uy\in E(G)$. It follows that $G[\{u,w,x,y\}]$ is a diamond, contradicting to the assumption that $(G,S,k)$ is reduced with respect to \texttt{Small Set Branching}. From $\{x\}\subseteq N_S(B)\cap C$, our claim follows. 
Notice that $\abs{C\cap B}\leq 2$ since an induced cycle can intersect with a clique in at most two vertices. Therefore, $(C\setminus B)\cup\{b\}$ has at least four vertices. Also $G'[(C\setminus B)\cup \{b\}]$ is an induced cycle as no chord can be added in the construction of $G'$ from $G$. This contradicts to the assumption that $G'-\tilde{S}'$ is a block graph. This completes the proof of the lemma.
\end{proof}

\begin{LEM}\label{lem:disjointalgo}
Given an instance $(G,S,k)$ to \disjointBLOC\ with $n=|V(G)|$, the algorithm {\bf Block}$(G,S,k)$ correctly returns a solution or outputs {\sc No} in time $O(3^{k+\ell}\cdot n^6)$.
\end{LEM}
\begin{proof}
The correctness of the algorithm is discussed above. We show that {\bf Block}$(G,S,k)$ has the claimed running time. The recursive execution of {\bf Block}$(G,S,k)$ can be depicted as a search tree ${\cal T}$, where each tree node corresponds to a call of the procedure {\bf Block}. It is easy to verify that {\bf Block}$(G,S,k)$ takes $O(n^5)$-time at each tree node: testing whether an $n$-vertex graph is a block graph can be done in time $O(n^2)$, and at line~\ref{line:ifsmallobs} there can be at most $O(n^3)$ such tests. Therefore, it suffices to bound the size of the search tree in order to establish the running time. For an instance $(G,S,k)$, we associate a measure $k+\ell$, where $\ell$ is the number of connected components in $G[S]$. Whenever {\bf Block}$(G,S,k)$ corresponds to a branching node in ${\cal T}$ (i.e. having at least two children), in each branching either $k$ or $\ell$ strictly decreases by at least 1. As $k+\ell \geq 0$ at any tree node, the number of branching nodes in any path from the root to a leaf is at most $k+\ell$. This bounds the number of leaves in ${\cal T}$ by $3^{k+\ell}$. The length of a longest path in ${\cal T}$ is at most $n+k+\ell$ since each recursive call either decrease $k+\ell$, or reduces the number of vertices by applying \texttt{Bypass Rule}. Therefore, the size of ${\cal T}$ is at most $O(3^{k+\ell}\cdot n)$ and {\bf Block}$(G,S,k)$ runs in time $O(3^{k+\ell}\cdot n^6)$. 
\end{proof}

Finally, to solve \BLOC, we apply the standard iterative compression technique. 
Together with the algorithm {\bf Block} for \disjointBLOC\ and its analysis given in Lemma~\ref{lem:disjointalgo}, we obtain an FPT algorithm stated in Theorem~\ref{thm:main2}. 

\begin{proof}[Proof of Theorem~\ref{thm:main2}]
We apply the standard iterative compressing technique. The algorithm involves two-step reduction of \BLOC: we first reduce \BLOC\ to {\sc Compression} problem, which reduces to \disjointBLOC. 

Fix an arbitrary labeling $v_1,\ldots ,v_n$ of $V(G)$ and let $G_i$ be the graph $G[\{v_1,\ldots , v_i\}]$ for $1\leq i\leq n$. From $i=1$ up to $n$, we consider the following {\sc Compression} Problem for \BLOC: given a graph $G_i$ and $S_i\subseteq V(G_i)$ such that $G_i-S_i$ is a block graph and $\abs{S_i}\leq k+1$, we aim to find a set $S_i'\subseteq V(G_i)$ such that $G_i-S_i'$ is a block graph and $\abs{S_i'}\leq k$, if one exists, and output \NO\ otherwise. Since block graphs are closed under taking induced subgraphs, $(G,k)$ is a \YES-instance of \BLOC\ if and only if $(G_i,S_i)$ is a \YES-instance for the {\sc Compression} for all $i$, where $(G_i,S_i)$ is a legitimate instance. Hence, we can correctly output that $(G,S)$ is a \NO-instance of \BLOC\ if $(G_i,S_i)$ is a \NO-instance for some $i$. Moreover, if $S'_i$ is a solution to the $i$-th instance of {\sc Compression}, then $(G_{i+1},S'_i\cup\{v_{i+1})$ is a legitimate instance for the $(i+1)$-th instance of {\sc Compression}. 

Given an instance $(G,S)$ of {\sc Compression}, we enumerate all possible intersections $I$ of $S$ and a desired solution to $(G,S)$. For each guessed set $I$, we solve the instance $(G-I,S\setminus I,k-\abs{I})$ to \disjointBLOC\ using the algorithm {\bf Block}. Note that $(G,S)$ is a \YES-instance if and only if  $(G-I,S\setminus I,k-\abs{I})$ is a \YES-instance for some $I\subseteq S$. If $\tilde{S}$ is a solution to $(G-I,S\setminus I,k-\abs{I})$, then $\tilde{S}\cup I$ is a solution to $(G,S)$ for {\sc Compression}. Conversely, if there is a solution $\tilde{S}$ to $(G,S)$, for the set $I=\tilde{S}\cap S$ the instance $(G-I,S\setminus I, k-\abs{I})$ is \YES\ for \disjointBLOC. Therefore, using the algorithm {\bf Block} for \disjointBLOC, we can correctly solve \BLOC.

It remains to prove the complexity of the algorithm. Given an instance $(G,S)$, we guess at most ${k+1 \choose i}$ sets $I$ of size $i$ for each $1\leq i\leq k$, and solve the resulting instance $(G-I,S\setminus I,k-\abs{I})$ of \disjointBLOC\ in time $O(3^{k-i+\ell}\cdot n^6)=O(9^{k-i}\cdot n^6)$. Here we use the fact that the number of connected components in $G[S-I]$ is bounded by $\abs{S-I}$. Summing up, \BLOC\ can be solved by running an algorithm for {\sc Compression} at most $n$ times, which yields the claimed running time \[n\cdot \sum_{i=0}^{k}{k+1 \choose i}\cdot O(9^{k-i}\cdot n^6)= O(10^k\cdot n^7).\qedhere\]
\end{proof}

%

\end{document}